\documentclass[sigconf,nonacm]{acmart}

\setcopyright{none}
\renewcommand\footnotetextcopyrightpermission[1]{}

\usepackage{amsmath}
\usepackage{booktabs}
\usepackage{tabularx}
\usepackage{colortbl}
\usepackage[table]{xcolor}

\newtheorem{theorem}{Theorem}
\title[Remembering What Lasts, Adapting to What Changes]{Remembering What Lasts, Adapting to What Changes:\\
Better User Understanding for LLM-Enhanced Sequential Recommendation}

\author{Yu Hou}
\affiliation{%
  \institution{School of Mathematics and Computing (Computational Science and Engineering), Yonsei University}
  \city{Seoul}
  \country{Republic of Korea}}
\email{houyu@yonsei.ac.kr}

\begin{document}

\begin{abstract}
Large language model (LLM)-enhanced recommender systems leverage rich item semantics to support personalized recommendation. However, semantic representations alone do not determine which historical behaviors reflect persistent preferences and which mainly indicate recent interests, leaving an important aspect of user understanding unresolved. Recent advances in LLM inference show that newly available information can be used to refine the internal state during inference, thereby improving subsequent predictions. Inspired by this principle, we propose \textsf{\textbf{MATE}} (Memory Adaptation with Temporal Evidence), an adaptive user modeling framework for LLM-enhanced sequential recommendation. \textsf{\textbf{MATE}} first evaluates each newly observed interaction from two temporal perspectives: whether it is repeatedly supported by historical behaviors and whether it is consistent with recent interactions. The resulting temporal evidence controls the updates of two user-specific memories, where the \emph{long-term memory} conservatively preserves persistent preferences while the \emph{short-term memory} rapidly adapts to recent interests. For each recommendation, a recent-context representation \emph{dynamically} determines how strongly the two memories contribute to the current user representation. During offline training, next-item prediction is jointly optimized with temporal supervision, while during online adaptation, the shared model remains fixed and only the two user memories are updated from newly observed interactions. Experiments on MovieLens-10M, Amazon Luxury Beauty, and KuaiRec show that \textsf{\textbf{MATE}} improves mean NDCG@10 over the strongest external baseline by 7.0--13.2\%. Further analyses support its ability to adapt to recent interests while retaining useful information about recurring earlier preferences.
\end{abstract}


\maketitle

\section{Introduction}
\label{sec:introduction}

Large language models (LLMs) have recently emerged as a promising paradigm for recommender systems, benefiting from their strong language understanding, semantic modeling, and generalization capabilities~\cite{geng2022p5,li2023recformer,liao2024llara}. Unlike conventional recommendation methods that primarily rely on discrete item IDs and user--item interactions, language-model-enhanced recommenders can exploit rich textual information to construct more informative representations of items and user behaviors. Recent studies have explored different ways of incorporating such knowledge, including formulating recommendation as a language modeling task~\cite{geng2022p5}, learning semantic representations from item content~\cite{li2023recformer}, developing semantic item representations from different domains for LLM-based recommendation~\cite{hou2026unitok}, integrating collaborative signals with LLM representations~\cite{liao2024llara}, and reasoning over historical evidence for preference understanding~\cite{hou2026proverec}. These advances provide a stronger semantic foundation for personalized recommendation; however, richer item representations alone do not necessarily lead to a better understanding of how user interests evolve over time.

%
%
%
%

Accurately understanding users remains challenging because user preferences are inherently dynamic and may exhibit different temporal characteristics. Some preferences can persist over long periods, whereas others emerge from recent behaviors and may rapidly change. Sequential recommendation methods have been developed to capture behavioral dependencies over interaction histories~\cite{kang2018sasrec,li2020tisasrec}, while several studies further distinguish long- and short-term user interests~\cite{yu2019slirec,zheng2022clsr}. Although these approaches have demonstrated the importance of modeling multiple temporal signals, they mainly focus on how historical behaviors are encoded or combined into user representations. A less explored question is how each newly observed interaction should \emph{change} the user model itself. Treating all new behaviors with the same update mechanism may respond too slowly to emerging interests, or conversely, allow a small number of recent interactions to excessively affect useful preferences that have been repeatedly observed over a longer period.

\begin{figure}[t]
  \centering
  \includegraphics[width=\columnwidth]{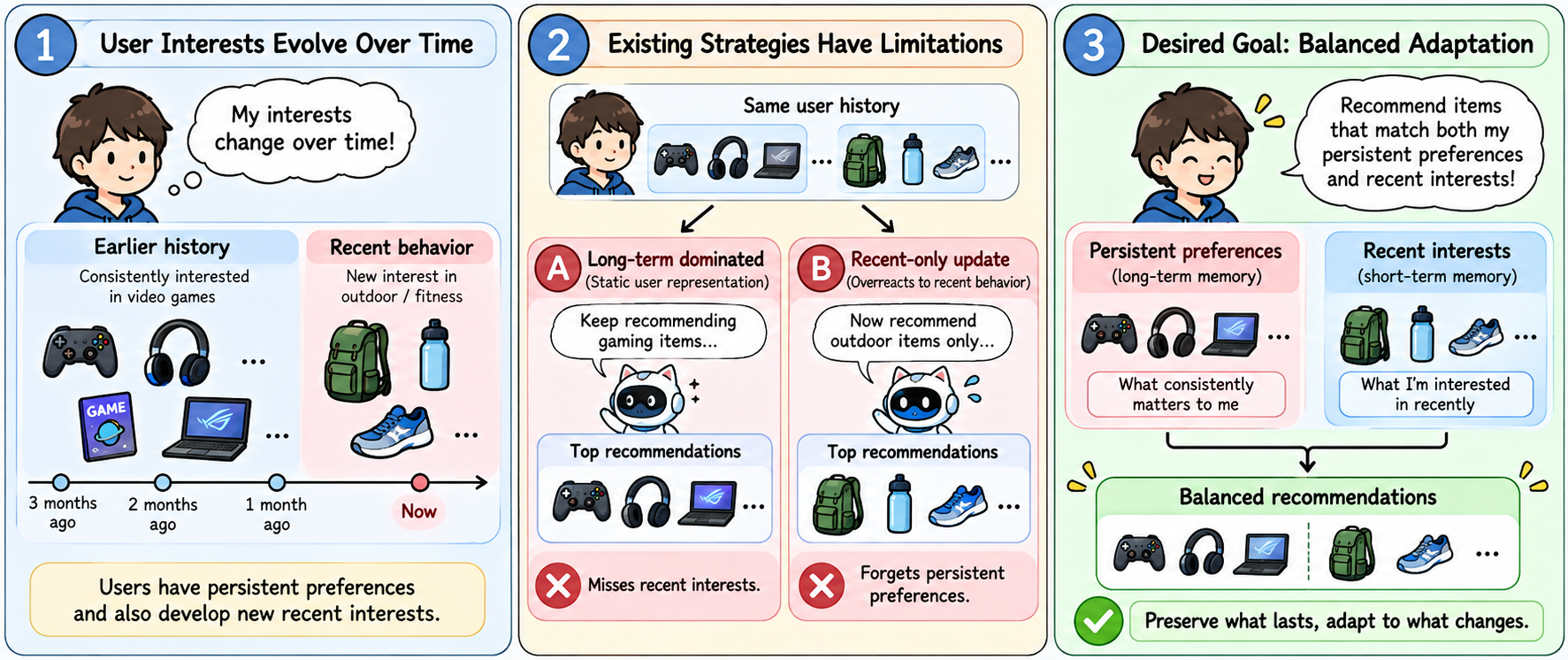}
  \Description{A motivation diagram contrasting persistent user preferences and newly emerging interests, illustrating the need to balance long-term memory with rapid adaptation to recent behavior.}
  \caption{\textbf{Motivation of MATE.} Users exhibit both persistent preferences and newly emerging interests over time. Existing strategies may overemphasize either long-term history or recent behavior, motivating a balanced user model that preserves what lasts while adapting to what changes.}
  \label{fig:motivation}
\end{figure}


%
%
%
%
%
%
%


Interestingly, recent advances in LLM inference provide a useful perspective for addressing this issue. Recent studies have shown that LLM performance can be further improved after offline training by allocating additional computation during inference, for example by searching over candidate reasoning paths, revising intermediate predictions, or adaptively refining the model state~\cite{snell2024scaling,sun2025tttlayers}. The key insight is that inference need not remain a single passive forward process: newly available information can be used to refine the internal state on which subsequent predictions are made. In-Place Test-Time Training (TTT) further demonstrates that such adaptation can be implemented by updating a subset of the model's existing weights using the observed context~\cite{feng2026inplacettt}. This principle naturally aligns with sequential recommendation: as new user interactions become observable, they provide additional evidence for refining the current understanding of the user, which can in turn improve subsequent recommendations. This motivates us to rethink user modeling as a continuous adaptive process rather than solely as a fixed encoding of the observed interaction history.


%
%
%

In this context, a natural question arising is: \textit{``How can we better understand evolving users by preserving their long-term preferences while adapting to recent interests?''} To answer this question, we outline the following two design challenges:
\begin{itemize}
    \item \textbf{C1. Identifying the temporal role of each interaction.} A newly observed interaction does not explicitly indicate whether it reflects a \textbf{persistent preference} or a \textbf{recent interest}. Simply treating earlier interactions as long-term and recent ones as short-term may be unreliable, since the same behavior can reinforce an existing preference while also being part of an active recent interest. Therefore, the model should determine how much evidence each observed interaction provides to the two temporal aspects.
    \item \textbf{C2. Adapting the user state without losing useful historical information.} Even after identifying the temporal role of an interaction, it remains unclear how this evidence should modify the user model. A uniform update may either respond too slowly to newly emerging interests or excessively overwrite useful preferences accumulated over a longer period. Moreover, the stored long-term and short-term information should not necessarily contribute equally to every recommendation.
\end{itemize}

To address these aforementioned challenges, we propose \textsf{\textbf{\underline{MATE}}}, a \textbf{\underline{M}}emory \textbf{\underline{A}}daptation with \textbf{\underline{T}}emporal \textbf{\underline{E}}vidence for better user understanding in LLM-enhanced sequential recommendation. an adaptive user modeling framework for LLM-enhanced sequential recommendation. Our key idea is to use \textbf{temporal evidence} to continuously update two \textbf{ser-specific memories} that capture long-term preferences and short-term interests, while dynamically determining how the two memories should contribute to the current recommendation.


%
%
%
%

% ...

\noindent\textbf{(\underline{Idea 1}): Temporal evidence for understanding user behaviors.} Instead of assigning interactions to long-term or short-term interests solely according to their timestamps, we evaluate each newly observed behavior from two complementary perspectives. We regard behaviors that repeatedly appear across separated historical periods as stronger evidence of \textbf{long-term preference}, while consistency with recent interactions provides evidence of \textbf{short-term interest}. Importantly, the two types of evidence are not mutually exclusive, allowing the same interaction to contribute to both user states with different strengths. In this way, the model captures the temporal meaning of each observed behavior before updating the user representation, thereby addressing \textbf{C1}.


%
%
%
%
%
%

% ...

\noindent\textbf{(\underline{Idea 2}): Adaptive user memories for continuously refining user understanding.} Inspired by recent adaptive LLM inference methods, we allow newly observed interactions to directly refine the user state instead of only serving as additional inputs to a fixed forward process. Specifically, we maintain a \textbf{long-term memory} and a \textbf{short-term memory} for each user, which follow different update and retention dynamics. The long-term memory changes more conservatively to preserve persistent behavioral patterns, whereas the short-term memory responds more rapidly to recent interests. Furthermore, the recent interaction context dynamically controls how much the two memories contribute to the current recommendation. This design separates \textbf{what should be remembered} from \textbf{what should be emphasized now}, thereby addressing \textbf{C2}.


%
%
%

Overall, our framework follows a simple adaptive process:
newly observed interactions are first evaluated through temporal evidence, then used to update the long-term and short-term user memories, and finally the two memories are selectively utilized for next-item recommendation. The shared recommendation model is learned offline and remains fixed during inference, while the user-specific memories continue to evolve only after real user interactions become observable.







Our main contributions are summarized as follows:

\begin{itemize}
    \item \textbf{Adaptive user understanding:} We connect temporal user modeling with recent advances in adaptive inference and investigate how newly observed interactions can continuously refine user understanding in LLM-enhanced sequential recommendation. Different from methods that primarily encode or combine temporal representations, we focus on how behavioral evidence should \textbf{modify the user state itself}.
    \item \textbf{Evidence-guided dual-memory modeling:} We propose \textbf{\textsf{MATE}}, which uses temporal evidence to control the updates of long-term and short-term user memories with different adaptation dynamics. A context-aware mechanism further determines how strongly the two memories contribute to each recommendation, enabling the model to preserve persistent preferences while adapting to recent changes.
    \item \textbf{Comprehensive evaluations:} Experiments on three real-world datasets show that \textsf{\textbf{MATE}} (a) improves mean NDCG@10 over the strongest external baseline by 7.0--13.2\%, (b) improves prediction of both recent interests and returns to recurring earlier themes, and (c) uses 1~MiB of adaptive memory per user, with 1.25--1.32~ms additional processing per event.
\end{itemize}

\section{Related Work}
\label{sec:related_work}

\subsection{LLM-Enhanced Sequential Recommendation}

Recent advances in language models have enabled recommender systems to exploit rich semantic information beyond conventional item IDs. Early language-based recommendation frameworks, such as P5, formulate recommendation tasks in a unified text-to-text paradigm~\cite{geng2022p5}. Other studies leverage textual item information to learn transferable or semantically enriched representations for sequential recommendation, including UniSRec~\cite{hou2022unisrec} and Recformer~\cite{li2023recformer}. More recent approaches further integrate collaborative signals with language models~\cite{liao2024llara}, employ LLMs for candidate ranking~\cite{yue2023llamarec}, or represent items using semantic identifiers for generative retrieval and scalable item tokenization~\cite{rajput2023tiger,hou2026unitok}. LLMs have also been used to reason over historical evidence for more explicit and verifiable preference modeling~\cite{hou2026proverec}.

These studies demonstrate that language models provide rich semantic representations for modeling items and user--item interactions. In this work, we use frozen language-model embeddings to represent item semantics and instead focus on how such semantic evidence can support a more adaptive understanding of users. Specifically, we study how persistent preferences and recently emerging interests can be accumulated and selectively utilized as the user's interaction history evolves.

\subsection{Temporal User Modeling in Sequential Recommendation}

Sequential recommendation has extensively studied how user preferences evolve over interaction histories. Representative models employ recurrent networks~\cite{hidasi2016gru4rec}, self-attention~\cite{kang2018sasrec}, bidirectional Transformers~\cite{sun2019bert4rec}, and explicit temporal intervals~\cite{li2020tisasrec} to capture sequential dependencies. Beyond modeling a single evolving representation, several studies explicitly distinguish preferences at different temporal scales. SLi-Rec models long- and short-term user preferences and adaptively aggregates them for recommendation~\cite{yu2019slirec}.

Particularly related to our work, CLSR disentangles long- and short-term interests using separate encoders and constructs history-derived proxies to supervise the two representations~\cite{zheng2022clsr}. This demonstrates that temporal interests benefit from explicit modeling rather than being implicitly compressed into a single user state. However, CLSR primarily separates and aggregates \emph{representations}. Our work investigates a different question: whether temporal evidence should also determine \emph{how user states are updated}. Instead of assigning interactions exclusively to one temporal representation, an observed interaction may update both persistent and recent memories with different strengths according to cross-period recurrence and recent behavioral consistency.

\subsection{Adaptive Inference for User Modeling}

Recent inference-time learning methods provide a complementary perspective in which observed context is not only encoded as input, but also used to update an internal model state. Test-Time Training (TTT) introduces self-supervised adaptation during inference~\cite{sun2020ttt}, while TTT layers formulate the hidden state itself as a learnable model that is continuously updated as new tokens are observed~\cite{sun2025tttlayers}. In-Place TTT further implements this idea within standard Transformer MLP blocks by treating existing projection matrices as fast weights and updating them with prediction-aligned local learning objectives~\cite{feng2026inplacettt}. Earlier fast-weight models similarly demonstrate that matrix-valued associative memories can store and retrieve recent contextual information~\cite{ba2016fastweights}.

The idea of adaptive inference has also been explored in recommendation. Related efforts have investigated advanced generative modeling and post-training optimization, including diffusion-based collaborative filtering~\cite{hou2024cfdiff} and reinforcement-learning-based fine-tuning of diffusion recommenders~\cite{hou2025refit}. More directly, TTT4Rec introduces TTT-style state updates into sequential recommendation, allowing newly observed interactions to modify the internal sequence state during inference~\cite{yang2024ttt4rec}. Related work further studies adaptation under distribution shift~\cite{yang2025dt3or}, temporal alignment for changing user interests~\cite{zhang2025t2arec}, and retrieval-based adaptation without parameter updates~\cite{tang2026read}.



Existing temporal recommenders model user interests at different time scales, while adaptive inference methods update user states from newly observed interactions. However, how temporal evidence should guide such adaptation remains underexplored. We therefore maintain two user-specific fast-weight memories for persistent and recent interests, using cross-period recurrence and recent-window consistency to control memory updates, while the current context determines their contributions to the next recommendation.

\section{Methodology}
\label{sec:methodology}


\subsection{Task Formulation}
\label{sec:task_formulation}

Let $\mathcal{U}$ and $\mathcal{I}$ denote the sets of users and items. Each item $i\in\mathcal{I}$ is associated with textual metadata
\begin{equation}
\mathbf{x}_i =
\{\text{title}_i,\text{categories}_i,\text{features}_i,
\text{description}_i\},
\end{equation}
which is encoded by a frozen language-model encoder $g(\cdot)$ to obtain a semantic item representation
\begin{equation}
\mathbf{e}_i=g(\mathbf{x}_i)\in\mathbb{R}^{d}.
\end{equation}

For each user $u$, we observe a chronological interaction sequence
\begin{equation}
\mathcal{S}_u=
{(i_{1},t_{1},r_{1}),\ldots,
(i_{T},t_{T},r_{T})},
\end{equation}
where $i_{t}$, $t_{t}$, and $r_{t}$ denote the interacted item, timestamp, and observed feedback, respectively. Given the observed history
\begin{equation}
\mathcal{H}_{u,t}={(i_1,t_1,r_1),\ldots,(i_t,t_t,r_t)},
\end{equation}
our goal is to predict the next item $i_{t+1}$.

Different from conventional sequential recommendation that derives the user state mainly from the current interaction sequence, we maintain two user-specific adaptive memories, long-term memory $\mathbf{M}^{L}_{u,t}$ and short-term memory $\mathbf{M}^{S}_{u,t}$, which are designed to capture persistent preferences and recent interests, respectively. Given $\mathcal{H}_{u,t}$ and the two memories, the recommender produces a user representation 
\begin{equation}
\mathbf{h}_{u,t}
=
f_{\theta}
\left(
\mathcal{H}_{u,t},
\mathbf{M}^{L}_{u,t},
\mathbf{M}^{S}_{u,t}
\right),
\end{equation}
and scores a candidate item $i$ by
\begin{equation}
s_{u,t}(i)=\mathbf{h}_{u,t}^{\top}\mathbf{e}_i .
\end{equation}

The shared parameters $\theta$ are learned offline, whereas the user-specific memories evolve sequentially from observed interactions. The recommendation objective is to predict the next item based only on the observed history and the current user memories. Formally,
\begin{equation}
\hat{i}_{u,t}
=
\arg\max_{i\in\mathcal{I}}
\left(
s_{u,t}(i) \mid
\mathcal{H}_{u,t},
\mathbf{M}^{L}_{u,t},
\mathbf{M}^{S}_{u,t}
\right).
\end{equation}

After the true interaction $i_{u,t}$ is observed, it is used to update the user memories for future recommendation:
\begin{equation}
\left(
\mathbf{M}^{L}_{u,t+1},
\mathbf{M}^{S}_{u,t+1}
\right)
=
\mathcal{A}
\left(
\mathbf{M}^{L}_{u,t},
\mathbf{M}^{S}_{u,t},
\mathcal{H}_{u,t},
i_{u,t}
\right).
\end{equation}
where $\mathcal{A}(\cdot)$ denotes the user-memory adaptation process. Therefore, $i_{u,t}$ is used to update the memories only after its prediction has been evaluated, preventing future information from leaking into the current recommendation\footnote{The detailed memory writing, reading, and learning mechanisms are introduced in the following sections.}.

\subsection{Overview}
\label{sec:method_overview}

Recent studies on LLM inference have shown that model performance can be further improved after pretraining by making better use of computation and information available at inference time. For example, inference-time scaling improves reasoning by allocating additional computation to explore, evaluate, or refine candidate solutions~\cite{snell2025scaling}, while test-time learning further allows the model state itself to adapt to the observed context~\cite{sun2025tttlayers}. These approaches improve prediction not by simply enlarging the pretrained model, but by extracting more useful information from the current input and using it to guide subsequent inference.

\begin{figure*}[t]
  \centering
  \includegraphics[width=\textwidth]{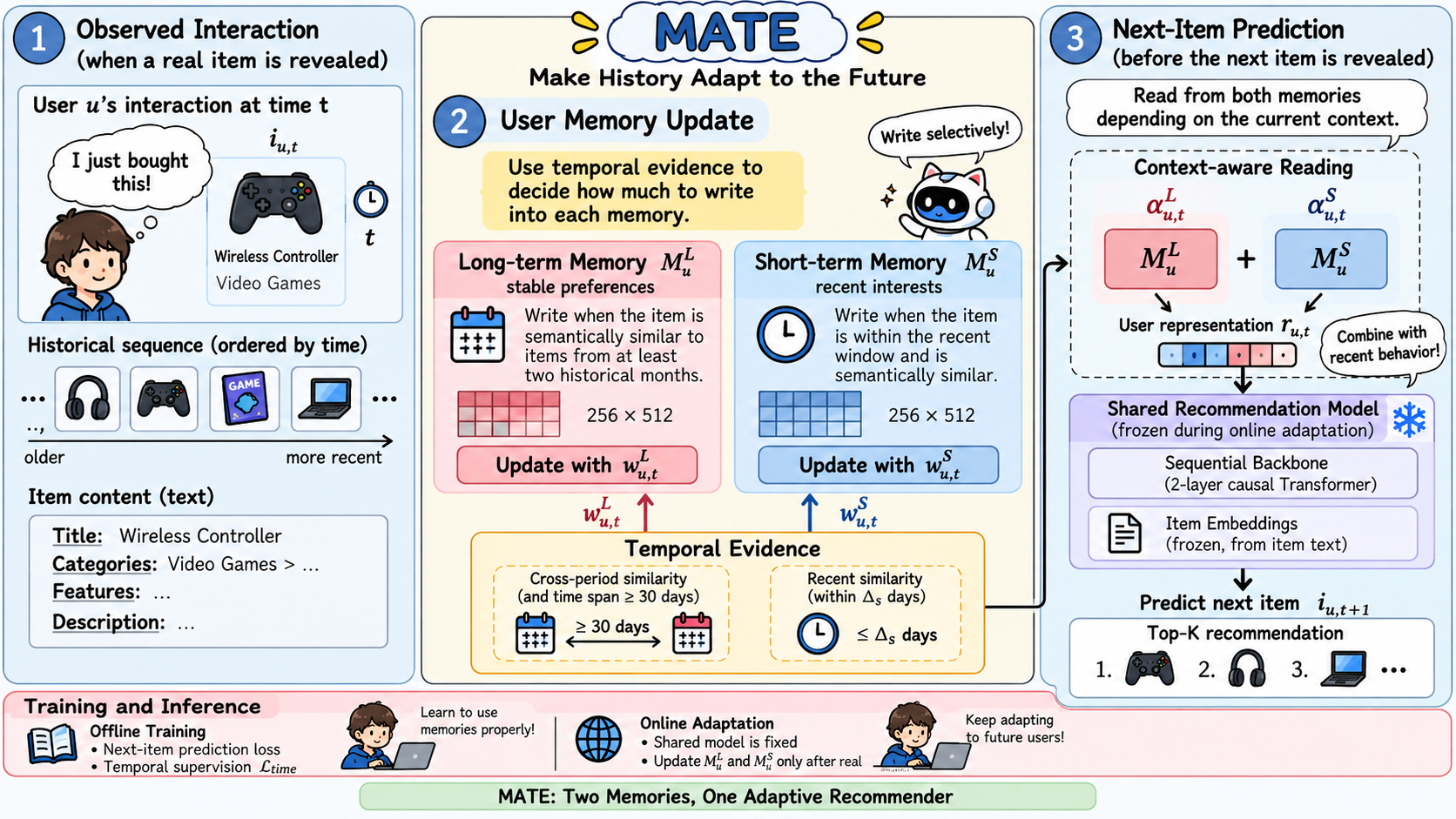}
  \Description{Overview of the MATE framework, showing temporal evidence computation, long- and short-term user-memory updates, context-conditioned memory reading, and next-item prediction.}
  \caption{\textbf{Overview of \textsf{\textbf{MATE}}.} A newly observed interaction is first evaluated by temporal evidence, which determines how strongly it updates the user’s long-term memory and short-term memory. For the next-item prediction, the model reads from both memories according to the current context, enabling MATE to preserve persistent preferences while adapting to recent interests.}
  \label{fig:overview}
\end{figure*}

Motivated by this idea, we propose a three-stage adaptive user modeling framework, illustrated in Fig.~\ref{fig:overview}. First, \textbf{Temporal Evidence} examines whether a newly observed behavior is repeatedly supported by the user's earlier history or is mainly consistent with recent interactions, providing evidence for long-term preference and short-term interest, respectively. Second, \textbf{Adaptive User Memory} uses these signals to update two user-specific memories: a long-term memory that preserves relatively stable preferences and a short-term memory that quickly adapts to recent interests. The shared recommendation model remains fixed during inference, while the two user memories evolve as new interactions are observed. Third, \textbf{Context-Aware Recommendation} determines how strongly the two memories should contribute under the current context and produces the user representation for next-item ranking. Frozen language-model embeddings provide semantic item representations throughout the framework. During offline training, the shared recommender and memory-control mechanisms are jointly learned, while during inference the shared parameters remain fixed and only the user-specific memories evolve after real interactions are observed.

\subsection{Temporal Evidence}
\label{sec:temporal_evidence}

As illustrated in Fig.~\ref{fig:overview}, the first stage determines how much a newly observed interaction reflects the user's long-term preference and short-term interest. Instead of dividing the interaction history into two fixed subsets, we evaluate each observed interaction using two independent temporal signals. This allows the same interaction to contribute to both user states when supported by the corresponding evidence.

Let $i_{u,t}$ denote the newly observed item for user $u$ at time $\tau_{u,t}$, and let $\mathbf{e}_{u,t}$ be its normalized semantic embedding. Since all item embeddings are $\ell_2$-normalized, the semantic similarity between two interacted items is computed as
\begin{equation}
\operatorname{sim}(i_{u,t},i_{u,j})
=
\mathbf{e}_{u,t}^{\top}\mathbf{e}_{u,j}.
\end{equation}
We regard two interactions as semantically related when
\begin{equation}
\operatorname{sim}(i_{u,t},i_{u,j}) \geq \delta,
\end{equation}
where $\delta$ is a similarity threshold.

\paragraph{Short-term evidence.}
Short-term evidence measures whether the current interaction is consistent with the user's recent behaviors. Let
\begin{equation}
\mathcal{R}_{u,t}
=
\{
j<t \mid
\tau_{u,t}-\tau_{u,j}\leq \Delta_S
\}
\end{equation}
denote the recent interaction window, where $\Delta_S$ controls the short-term time span. We compute the proportion of semantically related interactions within this window as
\begin{equation}
C^S_{u,t}
=
\frac{
\sum_{j\in\mathcal{R}_{u,t}}
\mathbb{I}
\left[
\mathbf{e}_{u,t}^{\top}\mathbf{e}_{u,j}
\geq \delta
\right]
}{
|\mathcal{R}_{u,t}|
}.
\end{equation}
The resulting short-term evidence is
\begin{equation}
w^S_{u,t}
=
\epsilon
+
(1-\epsilon)C^S_{u,t},
\end{equation}
where $\epsilon>0$ provides a minimum write strength for newly emerging interests. Therefore, an interaction can still enter the short-term state even when no similar recent behavior has yet been observed, while repeated recent interactions produce stronger evidence.

\paragraph{Long-term evidence.}
Long-term evidence focuses on behavioral patterns that repeatedly appear across different periods rather than within a single burst of activity. We first exclude the recent window and group the remaining historical interactions by calendar month. Let $\mathcal{M}_{u,t}$ denote the set of non-empty historical months before the recent window, and let $\mathcal{B}^{m}_{u,t}$ denote the interactions belonging to month $m\in\mathcal{M}_{u,t}$.

For each historical month, we define whether it contains at least one interaction semantically related to the current item:
\begin{equation}
z^{m}_{u,t}
=
\mathbb{I}
\left[
\max_{j\in\mathcal{B}^{m}_{u,t}}
\mathbf{e}_{u,t}^{\top}\mathbf{e}_{u,j}
\geq \delta
\right].
\end{equation}
The number of supporting months is then
\begin{equation}
N^{L}_{u,t}
=
\sum_{m\in\mathcal{M}_{u,t}} z^{m}_{u,t}.
\end{equation}

We assign non-zero long-term evidence only when the related behavior appears in at least $K_L$ distinct historical periods and spans a sufficient temporal range $\Delta_L$. Accordingly,
\begin{equation}
w^L_{u,t}
=
\begin{cases}
\displaystyle
\frac{N^{L}_{u,t}}
{|\mathcal{M}_{u,t}|},
&
N^{L}_{u,t}\geq K_L
\text{ and }
\operatorname{span}_{u,t}\geq\Delta_L,\\
0,
& \text{otherwise},
\end{cases}
\end{equation}
where $\operatorname{span}_{u,t}$ denotes the time interval between the current interaction and its earliest supporting historical interaction.

This design prevents a dense burst of similar behaviors within a single period from being immediately treated as a persistent preference. In contrast, repeatedly observed semantic patterns across separated periods provide stronger evidence that the behavior reflects a relatively stable user preference\footnote{In our implementation, we use $\Delta_S=7$ days, $K_L=2$, $\Delta_L=30$ days, $\delta=0.7$, and $\epsilon=0.2$ by default, while treating them as configurable hyperparameters.}.

The two evidence scores $w^L_{u,t}$ and $w^S_{u,t}$ are computed independently and are not constrained to sum to one. Thus, an interaction may simultaneously reinforce an existing long-term preference and reflect an active recent interest. Importantly, both scores are computed exclusively from interactions observed before or at the current step; no future interaction is used. The resulting evidence weights are passed to the adaptive user-memory module described in the next section.

\subsection{Adaptive User Memory}
\label{sec:adaptive_memory}


As illustrated in the second stage of Fig.~\ref{fig:overview}, the temporal evidence obtained in the previous section is used to update two user-specific memories: a long-term memory $\mathbf{M}^{L}_{u,t}$ for relatively persistent preferences and a short-term memory $\mathbf{M}^{S}_{u,t}$ for recently emerging interests. The two memories are maintained separately for each user and evolve only from observed interactions, while the shared recommendation model is common to all users.

For an observed interaction $i_{u,t}$, let $\mathbf{k}_{u,t-1}$ denote the internal representation produced by the final feed-forward layer before the current interaction is revealed,\footnote{$\mathbf{h}_{u,t}$ denotes the final user representation used for item scoring, while $\mathbf{k}_{u,t-1}$ is the intermediate feed-forward representation used for memory writing. In our implementation, $\mathbf{k}_{u,t-1}\in\mathbb{R}^{512}$ and $\mathbf{h}_{u,t}\in\mathbb{R}^{256}$.} and let $\mathbf{e}_{u,t}$ denote the semantic embedding of the newly observed item. The memory update associates the preceding user context $\mathbf{k}_{u,t-1}$ with the subsequently observed item $\mathbf{e}_{u,t}$.

For the long-term memory, we define
\begin{equation}
\ell^{L}_{u,t}
=
-
w^{L}_{u,t}
\mathbf{e}_{u,t}^{\top}
\mathbf{M}^{L}_{u,t-1}
\mathbf{k}_{u,t-1}
+
\frac{\gamma_L}{2}
\left|
\mathbf{M}^{L}_{u,t-1}
\right|_{F}^{2},
\end{equation}
where $w^{L}_{u,t}$ is the long-term evidence introduced in Section~\ref{sec:temporal_evidence}, and $\gamma_L$ controls the retention of previously stored long-term information.

Similarly, the short-term memory objective is
\begin{equation}
\ell^{S}_{u,t}
=
-
w^{S}_{u,t}
\mathbf{e}_{u,t}^{\top}
\mathbf{M}^{S}_{u,t-1}
\mathbf{k}_{u,t-1}
+
\frac{\gamma_S}{2}
\left|
\mathbf{M}^{S}_{u,t-1}
\right|_{F}^{2},
\end{equation}
where $w^{S}_{u,t}$ denotes the short-term evidence and $\gamma_S$ controls the retention of short-term information.

Applying one local gradient step with update rates $\eta_L$ and $\eta_S$ gives the long-term memory update
\begin{equation}
\mathbf{M}^{L}_{u,t}
=
(1-\eta_L\gamma_L)
\mathbf{M}^{L}_{u,t-1}
+
\eta_L
w^{L}_{u,t}
\mathbf{e}_{u,t}
\mathbf{k}_{u,t-1}^{\top},
\end{equation}
and the short-term memory update
\begin{equation}
\mathbf{M}^{S}_{u,t}
=
(1-\eta_S\gamma_S)
\mathbf{M}^{S}_{u,t-1}
+
\eta_S
w^{S}_{u,t}
\mathbf{e}_{u,t}
\mathbf{k}_{u,t-1}^{\top}.
\end{equation}

Each update contains two components. The first preserves previously accumulated information, while the second writes the newly observed interaction according to its corresponding temporal evidence. Therefore, interactions supported by repeated historical patterns produce stronger updates to the long-term memory, whereas interactions that are consistent with recent behaviors contribute more strongly to the short-term memory. Since $w^{L}_{u,t}$ and $w^{S}_{u,t}$ are computed independently, the same interaction may update both memories when it simultaneously reflects an established preference and an active recent interest.

The two memories also use different update dynamics. The long-term memory is updated more conservatively and retains previous information for longer, whereas the short-term memory reacts more quickly to newly observed behaviors and discounts older information more strongly. Thus, temporal evidence determines \emph{what} should be written, whereas the memory-specific update dynamics determine \emph{how quickly} the user state should change.

Importantly, $\mathbf{M}^{L}_{u,t}$ and $\mathbf{M}^{S}_{u,t}$ are user-specific adaptive memories rather than globally shared model parameters. During offline training, the shared parameters are optimized through the recommendation objective, while the two memories are generated sequentially by the above update rules. During inference, the shared parameters remain fixed, and only the user memories continue to evolve after real interactions are observed. The updated memories are then used for subsequent recommendations, as described in the next section.

\subsection{Context-Aware Recommendation}
\label{sec:context_aware}


After updating the long-term and short-term memories, the final stage in Fig.~\ref{fig:overview} determines how much information from each memory should contribute to the current recommendation. Although both memories summarize observed user behaviors, their relevance may vary across prediction steps. We therefore use the recent interaction context to dynamically control their contributions rather than combining them with fixed weights.

Given the observed history $\mathcal{H}_{u,t}$, we construct a recent-context representation $\mathbf{q}_{u,t}$ from the item embeddings within the same recent window $\Delta_S$ introduced in Section~\ref{sec:temporal_evidence}. Specifically,
\begin{equation}
\begin{aligned}
\mathbf{q}_{u,t}
&=
\sum_{j\in\mathcal{R}_{u,t}\cup\{t\}}
a_{u,t,j}\mathbf{e}_{u,j},\\
a_{u,t,j}
&=
\frac{
\exp\!\left[-\lambda(\tau_{u,t}-\tau_{u,j})\right]
}{
\sum_{k\in\mathcal{R}_{u,t}\cup\{t\}}
\exp\!\left[-\lambda(\tau_{u,t}-\tau_{u,k})\right]
}.
\end{aligned}
\end{equation}
where $\lambda$ controls the temporal decay, such that more recent interactions contribute more strongly to the current context.

Based on $\mathbf{q}_{u,t}$, a small gating network produces two independent memory-reading weights,
\begin{equation}
\left[
\alpha^{L}_{u,t},
\alpha^{S}_{u,t}
\right]
=
\sigma \left(
\operatorname{MLP}(\mathbf{q}_{u,t})
\right),
\end{equation}
where $\sigma(\cdot)$ denotes the sigmoid function. The two weights are independently estimated and are therefore not constrained to sum to one. This allows the model to emphasize either memory, both memories, or reduce their contributions when they are less relevant to the current context.

The long-term and short-term memories are incorporated into the final feed-forward transformation as
\begin{equation}
\mathbf{h}_{u,t}
=
\operatorname{Norm}
\left(
\mathbf{W}_{\mathrm{down}}\mathbf{k}_{u,t}
+
\alpha^{L}_{u,t}
\mathbf{M}^{L}_{u,t}\mathbf{k}_{u,t}
+
\alpha^{S}_{u,t}
\mathbf{M}^{S}_{u,t}\mathbf{k}_{u,t}
\right),
\end{equation}
where $\mathbf{W}_{\mathrm{down}}$ belongs to the shared model parameters $\theta$, while $\mathbf{M}^{L}_{u,t}$ and $\mathbf{M}^{S}_{u,t}$ are the user-specific memories updated in Section~\ref{sec:adaptive_memory}. $\operatorname{Norm}(\cdot)$ denotes the normalization used to obtain the final user representation $\mathbf{h}_{u,t}$.

The resulting user representation is matched with the frozen semantic embedding of each candidate item:
\begin{equation}
s_{u,t}(i)
=
\mathbf{h}_{u,t}^{\top}\mathbf{e}_{i},
\end{equation}
and candidate items are ranked according to $s_{u,t}(i)$.

It is important to distinguish memory writing from memory reading. The temporal evidence $w^{L}_{u,t}$ and $w^{S}_{u,t}$ determines how strongly an observed interaction updates the two memories, whereas $\alpha^{L}_{u,t}$ and $\alpha^{S}_{u,t}$ determine how strongly the stored long-term and short-term information is used for the current recommendation. Therefore, information can remain in a user memory without necessarily contributing equally to every subsequent prediction.

\subsection{Joint Learning and Online Adaptation}
\label{sec:learning}


The shared model parameters are learned offline, while the long-term and short-term memories are updated sequentially using the observed interactions of each user. During training, the model follows the same causal process as in inference: it first predicts the next item using the current memories, and only after the true interaction is revealed are the memories updated according to Section~\ref{sec:adaptive_memory}.

For next-item prediction, given the user representation $\mathbf{h}_{u,t}$, the recommendation loss is defined as
\begin{equation}
\mathcal{L}_{\mathrm{rec}}
=
-
\sum_{u,t}
\log
\frac{
\exp \left(s_{u,t}(i_{u,t+1})\right)
}{
\sum_{i\in\mathcal{C}_{u,t}}
\exp \left(s_{u,t}(i)\right)
},
\end{equation}
where $\mathcal{C}_{u,t}$ denotes the candidate set and $i_{u,t+1}$ is the next observed item.

To encourage the two memories to preserve information consistent with their temporal roles, we additionally construct a long-term target $\mathbf{p}^{L}_{u,t}$ and a short-term target $\mathbf{p}^{S}_{u,t}$ from the observed history. The former summarizes repeatedly supported historical behaviors, while the latter summarizes recent interactions. Let $\hat{\mathbf{p}}^{L}_{u,t}$ and $\hat{\mathbf{p}}^{S}_{u,t}$ denote the corresponding representations read only from $\mathbf{M}^{L}_{u,t}$ and $\mathbf{M}^{S}_{u,t}$. We define
\begin{equation}
\mathcal{L}_{\mathrm{time}}
=
\left|
\hat{\mathbf{p}}^{L}_{u,t}
-
\mathbf{p}^{L}_{u,t}
\right|_2^2
+
\left|
\hat{\mathbf{p}}^{S}_{u,t}
-
\mathbf{p}^{S}_{u,t}
\right|_2^2,
\end{equation}
where a term is omitted when the corresponding temporal target cannot be reliably constructed.

The overall offline training objective is
\begin{equation}
\mathcal{L}
=
\mathcal{L}_{\mathrm{rec}}
+
\lambda\mathcal{L}_{\mathrm{time}},
\end{equation}
where $\lambda$ controls the temporal supervision. The shared parameters, including the sequential recommender and the memory-reading components, are optimized with this objective. The user memories themselves are not globally shared trainable parameters; instead, they are generated and updated sequentially through the memory-update equations in Section~\ref{sec:adaptive_memory}.

During online adaptation, all shared parameters are fixed. For each user, $\mathbf{M}^{L}_{u,t}$ and $\mathbf{M}^{S}_{u,t}$ continue to evolve whenever a new real interaction is observed. Thus, newly observed behaviors can influence subsequent recommendations without modifying the shared recommendation model, while future interactions are never used to update the memories before they are predicted.

\subsection{Theoretical Analysis}
\label{sec:theoretical_analysis}

We analyze two properties of \textsf{MATE}: how temporal evidence restricts
sporadic long-term writes, and how an individual memory update
influences subsequent user states. The results characterize the
memory mechanism under explicit assumptions, rather than guarantee
improvements in recommendation accuracy. Proofs are provided in
Appendix~\ref{app:mate_theory}.

\begin{theorem}[Selective long-term writing]
\label{thm:selective_writing}
Fix a user and a newly observed item. Let
$n=|\mathcal{M}_{u,t}|\geq 2$ and $K_L=2$.
Conditional on this item and the available historical month bins,
suppose that the monthly support indicators $z^m_{u,t}$ are
independent, with $\Pr(z^m_{u,t}=1)\leq p$ for every month,
where $p\in[0,1]$. Then
\begin{equation}
    \mathbb{E}[w^L_{u,t}]
    \leq
    p\bigl[1-(1-p)^{n-1}\bigr]
    \leq
    (n-1)p^2.
    \label{eq:long_write_bound}
\end{equation}
\end{theorem}

For a fixed number of historical months, the bound is quadratic
in $p$ when individual monthly matches are rare. Thus, requiring
repeated support limits the expected strength of sporadic
long-term writes. The additional time-span requirement can only
reduce this strength. Meanwhile, $w^S_{u,t}\geq\epsilon$ allows
new behaviors to contribute to short-term memory without
established long-term support.

\begin{theorem}[Stability of memory adaptation]
\label{thm:memory_stability}
Let $\rho_L=1-\eta_L\gamma_L$ and
$\rho_S=1-\eta_S\gamma_S$, with
$0<\rho_S<\rho_L<1$.
Assume radial Frobenius-norm clipping of each write increment
at $c>0$ and each updated memory at $B>0$.

Compare MATE with a reference trajectory
$\widetilde{\mathbf{M}}^L,\widetilde{\mathbf{M}}^S$
that omits only the additive write for $i_{u,t}$, retaining
the decay at that step. Both trajectories have identical memories
before $t$ and use the same observed history, context representations,
and write evidence thereafter. For any integer $s\geq0$,
\begin{equation}
\begin{aligned}
\|\mathbf{M}^L_{u,t+s}
-\widetilde{\mathbf{M}}^L_{u,t+s}\|_F
&\leq
\rho_L^s
\min\{c,\eta_Lw^L_{u,t}
\|\mathbf{k}_{u,t-1}\|_2\}.
\end{aligned}
\label{eq:long_influence_bound}
\end{equation}
\begin{equation}
\begin{aligned}
\|\mathbf{M}^S_{u,t+s}
-\widetilde{\mathbf{M}}^S_{u,t+s}\|_F
&\leq
\rho_S^s
\min\{c,\eta_Sw^S_{u,t}
\|\mathbf{k}_{u,t-1}\|_2\}.
\end{aligned}
\label{eq:short_influence_bound}
\end{equation}
\end{theorem}

The influence of an individual write cannot grow under identical
subsequent inputs. Temporal evidence controls its initial magnitude,
while the memory-specific retention factors govern its subsequent
influence bound. In particular, $w^L_{u,t}=0$ introduces no new
association into long-term memory, while the short-term branch
can respond to the same event. Here, $s$ counts memory updates,
not elapsed days.

\section{Experiments}
\label{sec:experiments}

We evaluate \textsf{MATE} from the perspectives of recommendation accuracy, component effectiveness, user adaptation, computational cost, and individual recommendation behavior. The experiments address the following research questions (\textbf{RQs}):
\begin{itemize}
    \item \textbf{RQ1:} \textbf{Overall Effectiveness}. How does \textsf{MATE} perform compared with representative sequential recommendation competitors?
    \item \textbf{RQ2:} \textbf{Component Effectiveness}. How do temporal evidence, dual user memories, temporal supervision, and context-aware reading affect performance?
    \item \textbf{RQ3:} \textbf{User Adaptation.} Does \textsf{MATE} retain useful long-term information while adapting to recent interests?
    \item \textbf{RQ4:} \textbf{Efficiency Analysis.} What computational and storage costs accompany user-memory adaptation?
    \item \textbf{RQ5:} \textbf{Qualitative User Understanding.} How do memory updates and reading affect recommendations for individual users?
\end{itemize}

\subsection{Experimental Setup}
\label{sec:exp_setup}


\paragraph{Datasets and preprocessing}
We use MovieLens-10M\footnote{\url{https://grouplens.org/datasets/movielens/10m/}}, Amazon Luxury Beauty 2018\footnote{\url{https://cseweb.ucsd.edu/~jmcauley/datasets/amazon_v2/}}, and KuaiRec 2.0\footnote{\url{https://zenodo.org/records/18164998}}. Positive interactions are ratings of at least four for MovieLens and Luxury Beauty, and a watch ratio greater than two for KuaiRec. We retain the first positive event for each user--item pair and remove events without item metadata. Item text contains movie titles and genres; product titles, categories, features, and descriptions; or video captions, categories, cover text, and topic tags. Review text, MovieLens user tags, aggregate ratings, and feedback counts are excluded. Timestamps represent rating, review, or video-interaction times; Luxury Beauty records are not interpreted as purchases.

After filtering, we split each dataset globally by time into approximately 80\% training, 10\% validation, and 10\% test, keeping equal timestamps in the same partition. We then select at most 1,000 users using only training-period activity, requiring at least 10, 8, and 10 positive training interactions on the three datasets, respectively. The user cohorts are fixed across training seeds. Consequently, their partition counts need not follow the global 80/10/10 ratio. Table~\ref{tab:dataset} reports the resulting statistics. The recent-context representation is derived from observed history; no explicit user query is supplied.

\begin{table*}[t]
    \centering
    \small
    \setlength{\tabcolsep}{4pt}
    \caption{Statistics of the selected user cohorts. The user count includes all selected users, not only users with test events. Train counts retained positive interactions; validation and test count scored next-item targets. Candidates comprise the full static metadata catalog. Coverage is the proportion of targets contained in this catalog.}
    \label{tab:dataset}
    \resizebox{0.7\textwidth}{!}{%
    \begin{tabular}{lrrrrrcc}
        \toprule
        Dataset & Users & Candidates & Train & Validation & Test
        & Val. coverage & Test coverage \\
        \midrule
        MovieLens-10M      & 1,000 & 10,681 & 76,016 & 1,651  & 913    & 100\% & 100\% \\
        Luxury Beauty 2018 &   948 & 12,111 & 10,789 & 981    & 491    & 100\% & 100\% \\
        KuaiRec 2.0        & 1,000 & 10,728 & 98,378 & 12,647 & 12,421 & 100\% & 100\% \\
        \bottomrule
    \end{tabular}}
\end{table*}

\paragraph{Candidate and evaluation protocol}
All methods rank the same fixed metadata catalog exactly, without sampled evaluation negatives. Previously interacted items are masked, and evaluation targets are never added to the catalog. Validation begins by replaying training history; test evaluation independently replays training and validation history. Shared parameters remain frozen. Each target is predicted and scored before it is revealed and used to update memory. States are isolated by user, events are not written twice, and generated recommendations are never treated as feedback.

\paragraph{Baselines}
We compare MATE with GRU4Rec~\cite{hidasi2016gru4rec}, SASRec~\cite{kang2018sasrec}, TiSASRec~\cite{li2020tisasrec}, CLSR~\cite{zheng2022clsr}, and TTT4Rec~\cite{yang2024ttt4rec}. We also include Oracle4Rec~\cite{xia2025oracle4rec} and PCTM~\cite{petrov2026pctm}. CLSR uses long--short fusion without the original contrastive objective. PCTM uses training-only transition statistics and no item embeddings. Base uses \textsf{MATE}'s backbone without adaptive memories, while Single uses one user-memory matrix. All methods share the data split, candidate catalog, and evaluation protocol.

\paragraph{Evaluation metrics}
We report event-averaged Recall@$K$ and NDCG@$K$ for $K\in\{5,10\}$. Each event has one relevant next item, so Recall equals Hit Rate. Overall metrics include uncovered targets with zero credit; covered-only metrics use targets in the catalog. Coverage is 100\% here, so these two results coincide. Stochastic methods use 3 seeds with means and sample standard deviations reported. PCTM is deterministic and uses one validation-selected fit. Paired uncertainty estimates resample users with their event sequences intact.

\paragraph{Implementation details}
Item text is encoded offline with \texttt{qwen3-embedding:\allowbreak 4b} through Ollama into frozen, 256-dimensional, $\ell_2$-normalized vectors. \textsf{MATE} uses two pre-normalized causal Transformer layers with four heads, hidden size 256, FFN size 512, and a context of at most 64 interactions.
Each user has two $256\times512$ memory matrices. During evaluation, memory replay covers the full observed history, beyond the context window. The recent window is seven days, the similarity threshold is 0.7, and the short-write floor is $\epsilon=0.2$. Long writes require matches in at least two older calendar months and at least 30 days between the earliest match and the revealed event. We use $(\eta_L,\gamma_L)=(0.03,0.001)$, $(\eta_S,\gamma_S)=(0.15,0.1)$, score temperature 0.07, and temporal-loss coefficient 0.1. Write increments and memory matrices
are norm-limited.

Internal variants use AdamW with learning rate $3\times10^{-4}$, weight decay 0.01, and training blocks of eight events. Within each seed, they share a two-epoch Base warm start, followed by five continuation epochs per method, including Base. External neural adapters train for seven epochs with 63 sampled training negatives and method-specific objectives and batch settings; equal epochs do not imply equal optimizer steps or compute. KuaiRec training uses at most the first 200 events per user, while evaluation retains the full observed history. Hyperparameter and checkpoint selection use validation NDCG@10 only. Experiments use PyTorch 2.14.0 on a Mac mini with an M4 chip and 16~GB unified memory: internal models run on MPS and external baselines on CPU.

\subsection{Overall Effectiveness}
\label{sec:overall}

To answer \textbf{RQ1}, we compare \textsf{MATE} with sequential baselines
under the same chronological split, fixed candidate catalog,
observed-history mask, and exact full-ranking protocol. Hyperparameters
and checkpoints are selected by validation NDCG@10, with the test split
reserved for final evaluation. Base, Single, and \textsf{MATE} use a
two-epoch warm-start selection budget followed by five continuation
epochs per method. External neural adapters train for seven epochs;
however, different objectives and batch settings mean that equal epoch
counts do not imply equal computation. Table~\ref{tab:main_results}
summarizes the overall results.

\textsf{MATE} achieves the best mean performance across all reported
metrics on all three datasets. Relative to the strongest external
baseline on NDCG@10, its improvements are approximately 8.6\% on
MovieLens-10M, 13.2\% on Luxury Beauty, and 7.0\% on KuaiRec. Its lead
at both cutoffs indicates that the benefit is not restricted to one
recommendation-list length. The strongest alternatives differ across
datasets: TTT4Rec leads the external baselines on MovieLens-10M,
whereas PCTM leads on Luxury Beauty. On KuaiRec, TTT4Rec is the strongest
external method on three metrics, while Oracle4Rec has the highest
external NDCG@5. Thus, \textsf{MATE}'s advantage holds against different
competitive approaches rather than a single consistently weak comparator.

Within the shared-backbone family, \textsf{MATE} improves NDCG@10 over
Base by approximately 93.8\%, 46.6\%, and 28.1\% on MovieLens-10M,
Luxury Beauty, and KuaiRec, respectively. Single also outperforms Base
on all reported metrics, suggesting that user-memory adaptation adds
useful predictive information beyond the shared backbone and frozen
item semantics. \textsf{MATE} further improves on Single across all
three datasets. This pattern is consistent with the motivation for
maintaining persistent and recent information separately, rather than
writing every interaction into one undifferentiated memory. Nevertheless,
Single has half the adaptive-memory capacity, so this comparison cannot
isolate temporal decomposition from additional storage. The equal-capacity
ablations provide the more direct component comparisons; the overall
results establish the performance ordering without attributing the gains
to any single mechanism.

\begin{table*}[t]
    \centering
    \small
    \setlength{\tabcolsep}{4pt}
    \caption{Overall recommendation results. The best result in each column is bold, and the second-best is underlined. For the MovieLens-10M, Luxury Beauty, and KuaiRec datasets, a paired \textit{t}-test is conducted, yielding p-values of 0.0129, 0.0203, and 0.0211, respectively, all below the threshold of 0.05, indicating statistically significant results.}
    \label{tab:main_results}
    \resizebox{\textwidth}{!}{%
    \begin{tabular}{l*{12}{c}}
        \toprule
        & \multicolumn{4}{c}{MovieLens-10M}
        & \multicolumn{4}{c}{Luxury Beauty}
        & \multicolumn{4}{c}{KuaiRec} \\
        \cmidrule(lr){2-5}\cmidrule(lr){6-9}\cmidrule(lr){10-13}
        Method & R@5 & R@10 & N@5 & N@10
               & R@5 & R@10 & N@5 & N@10
               & R@5 & R@10 & N@5 & N@10 \\
        \midrule
        GRU4Rec
        & 0.0099 & 0.0179 & 0.0070 & 0.0096
        & 0.0075 & 0.0170 & 0.0041 & 0.0071
        & 0.0128 & 0.0155 & 0.0115 & 0.0124 \\
        SASRec
        & 0.0073 & 0.0117 & 0.0047 & 0.0061
        & 0.0027 & 0.0143 & 0.0013 & 0.0049
        & 0.0153 & 0.0171 & 0.0138 & 0.0149 \\
        TiSASRec
        & 0.0095 & 0.0146 & 0.0060 & 0.0076
        & 0.0068 & 0.0122 & 0.0032 & 0.0050
        & 0.0136 & 0.0152 & 0.0124 & 0.0137 \\
        CLSR
        & 0.0058 & 0.0084 & 0.0031 & 0.0039
        & 0.0088 & 0.0129 & 0.0052 & 0.0066
        & 0.0151 & 0.0169 & 0.0146 & 0.0158 \\
        TTT4Rec
        & \underline{0.0131} & \underline{0.0226}
        & \underline{0.0085} & \underline{0.0116}
        & 0.0136 & 0.0177 & 0.0069 & 0.0083
        & \underline{0.0211} & \underline{0.0230}
        & 0.0191 & \underline{0.0213} \\
        Oracle4Rec
        & 0.0117 & 0.0204 & 0.0076 & 0.0104
        & 0.0095 & 0.0177 & 0.0046 & 0.0072
        & 0.0208 & 0.0227 & \underline{0.0192} & 0.0205 \\
        PCTM
        & 0.0099 & 0.0164 & 0.0070 & 0.0092
        & \underline{0.0143} & \underline{0.0183}
        & \underline{0.0102} & \underline{0.0114}
        & 0.0188 & 0.0198 & 0.0183 & 0.0195 \\
        \midrule
        Base
        & 0.0073 & 0.0135 & 0.0045 & 0.0065
        & 0.0093 & 0.0156 & 0.0067 & 0.0088
        & 0.0177 & 0.0182 & 0.0163 & 0.0178 \\
        Single
        & 0.0114 & 0.0189 & 0.0084 & 0.0093
        & 0.0103 & 0.0164 & 0.0087 & 0.0103
        & 0.0202 & 0.0227 & 0.0184 & 0.0201 \\
        \rowcolor{lightgray!30}
        \textsf{MATE}
        & \textbf{0.0149} & \textbf{0.0234}
        & \textbf{0.0105} & \textbf{0.0126}
        & \textbf{0.0157} & \textbf{0.0207}
        & \textbf{0.0119} & \textbf{0.0129}
        & \textbf{0.0228} & \textbf{0.0253}
        & \textbf{0.0204} & \textbf{0.0228} \\
        \bottomrule
    \end{tabular}}
\end{table*}

\subsection{Component Effectiveness}
\label{sec:ablation}

For \textbf{RQ2}, we examine memory writing, temporal supervision, dynamic reading, and continued online adaptation through controlled variants of \textsf{MATE}. \textbf{Base} retains the shared recommendation backbone and frozen item embeddings but has no user memory. \textbf{Single} replaces the long- and short-term matrices with one matrix, uses a unit write coefficient for every revealed interaction, and reads both temporal auxiliary targets from that same memory. \textbf{w/o Temporal Evidence} retains the two matrices, their separate update rates and decay settings, temporal supervision, and dynamic reading, but sets $w^L=w^S=1$. It therefore removes evidence-dependent write weighting rather than removing memory adaptation.

\textbf{w/o Temporal Supervision} sets the coefficient of $\mathcal{L}_{\mathrm{time}}$ to zero while retaining evidence-weighted writes and dynamic reading. \textbf{Fixed Reading} retains both memories, temporal evidence, and supervision, but replaces the learned reading gates with $\alpha^L=\alpha^S=0.5$ during training and testing. \textbf{Frozen Online Memory} reuses the full \textsf{MATE} checkpoint: memories are constructed by replaying training and validation history, then held fixed during testing. Newly revealed test interactions still update the observed history, recent context, and candidate mask. \textbf{MATE} retains all components and continues writing after each real test interaction is revealed. Training-time variants are retrained under the same budget, with hyperparameters and checkpoints selected using validation data. All test predictions precede target revelation.

Table~\ref{tab:ablation} shows that \textsf{MATE} outperforms every retrained variant on both metrics across all three datasets. Removing temporal evidence or temporal supervision lowers NDCG@10 by approximately 4.4--19.0\% and 6.6--25.4\%, respectively, relative to \textsf{MATE}. Fixed Reading produces the largest reduction among the retrained two-memory variants, lowering NDCG@10 by 32.5--56.3\%. These comparisons support evidence-dependent writing, memory-specific supervision, and context-dependent reading under the evaluated settings. Single also trails \textsf{MATE}, but its smaller state capacity prevents attributing that difference solely to temporal decomposition. The other retrained memory ablations retain the same two-matrix capacity.

Frozen Online Memory matches \textsf{MATE} on both metrics at the reported precision across all three datasets. Thus, these results do not demonstrate an additional Top-10 benefit from continued test-time writing. This control remains relevant because it separates the use of adapted memory from the benefit of further online updates; equality of rounded means does not establish identical predictions or event-level outcomes.

\begin{table*}[t]
    \centering
    \small
    \setlength{\tabcolsep}{4pt}
    \caption{Ablation results averaged over three seeds. Bold marks the best mean in each column, including ties. Memory denotes adaptive matrices stored in float32, excluding model weights and cached history. Frozen Online Memory reuses the MATE checkpoint without retraining.}
    \label{tab:ablation}
    \resizebox{\textwidth}{!}{%
    \begin{tabular}{lccccccc}
        \toprule
        & \multicolumn{2}{c}{MovieLens-10M}
        & \multicolumn{2}{c}{Luxury Beauty}
        & \multicolumn{2}{c}{KuaiRec} & \\
        \cmidrule(lr){2-3}\cmidrule(lr){4-5}\cmidrule(lr){6-7}
        Variant & R@10 & N@10 & R@10 & N@10 & R@10 & N@10
        & Memory (MiB/user) \\
        \midrule
        Base
        & 0.0135 & 0.0065
        & 0.0156 & 0.0088
        & 0.0182 & 0.0178 & 0 \\
        Single
        & 0.0189 & 0.0093
        & 0.0164 & 0.0103
        & 0.0227 & 0.0201 & 0.5 \\
        w/o Temporal Evidence
        & 0.0202 & 0.0102
        & 0.0187 & 0.0117
        & 0.0236 & 0.0218 & 1.0 \\
        w/o Temporal Supervision
        & 0.0197 & 0.0094
        & 0.0184 & 0.0113
        & 0.0241 & 0.0213 & 1.0 \\
        Fixed Reading
        & 0.0110 & 0.0055
        & 0.0144 & 0.0079
        & 0.0166 & 0.0154 & 1.0 \\
        Frozen Online Memory
        & \textbf{0.0234} & \textbf{0.0126}
        & \textbf{0.0207} & \textbf{0.0129}
        & \textbf{0.0253} & \textbf{0.0228} & 1.0 \\
        \rowcolor{lightgray!30}
        \textsf{MATE}
        & \textbf{0.0234} & \textbf{0.0126}
        & \textbf{0.0207} & \textbf{0.0129}
        & \textbf{0.0253} & \textbf{0.0228} & 1.0 \\
        \bottomrule
    \end{tabular}}
\end{table*}

\subsection{User Adaptation}
\label{sec:adaptation_analysis}

For \textbf{RQ3}, we test whether \textsf{MATE} can adapt to a recent interest while retaining useful earlier information. Predicting recent behavior alone could reward excessive recency bias; predicting actual returns to an earlier theme additionally tests retention. We identify observed histories with a recurring earlier theme A followed by a distinct recent theme B. A recurs in at least two older calendar months spanning 30 days. B forms a consecutive suffix within seven days and accounts for at most 20\% of older labeled interactions. Themes are fixed metadata groups or training-fitted semantic groups used only for analysis, never additional model inputs. After 1, 2, 3, or 5 observed B interactions, \textbf{Continue B} evaluates a real next item belonging to B, testing adaptation. \textbf{Return to A} evaluates a real next item belonging to A, testing prediction of a return rather than merely assigning high scores to A items. Eligibility uses only the observed prefix; the target theme labels the evaluation outcome. Predictions precede target revelation and memory updating.

Table~\ref{tab:adaptation} pools the preset B-count positions. \textsf{MATE} exceeds Base and Single in mean NDCG@10 for both outcomes on all three datasets. Better Continue B performance supports adaptation to recent behavior, while better Return to A performance suggests that this responsiveness does not sacrifice prediction of recurring earlier interests in these cohorts. Together, these outcomes support the intended adaptation--retention balance, without isolating the long-term matrix's causal contribution. Small cohorts and pooled B counts limit conclusions about individual users or adaptation at each streak length.

Frozen matches \textsf{MATE} at the reported precision, so these results do not establish an additional benefit from test-period memory writes. Table~\ref{tab:memory_analysis} adds diagnostics over all test events. Long writes trigger in 62--65\% of MovieLens/Luxury events but only 6.1\% of KuaiRec events; mean short-memory read contributions exceed long-memory contributions on every dataset. These statistics show dataset-dependent memory usage, but activity alone does not establish useful retention. The comparisons are descriptive, without a statistical-significance claim.

\begin{table*}[t]
\centering
\small
\setlength{\tabcolsep}{4pt}
\caption{NDCG@10 for recent-theme continuation (Continue B) and actual returns to the recurring earlier theme (Return to A). Values are means over three training seeds, pooling observed-B counts 1, 2, 3, and 5. $N/U$ denotes unique evaluated events/users per seed. Bold marks the best mean, including ties. Frozen disables test-period memory writes.
Candidate coverage is 100\%.}
\label{tab:adaptation}
\resizebox{0.7\textwidth}{!}{%
\begin{tabular}{llrcccc}
\toprule
Dataset & Next real target & $N/U$ & Base & Single & MATE & Frozen \\
\midrule
MovieLens-10M & Continue B & 74/21
& 0.0055 & 0.0068 & \textbf{0.0146} & \textbf{0.0146} \\
MovieLens-10M & Return to A & 149/24
& 0.0032 & 0.0037 & \textbf{0.0049} & \textbf{0.0049} \\
\addlinespace
Luxury Beauty & Continue B & 51/44
& 0.0067 & 0.0109 & \textbf{0.0141} & \textbf{0.0141} \\
Luxury Beauty & Return to A & 49/44
& 0.0041 & 0.0076 & \textbf{0.0103} & \textbf{0.0103} \\
\addlinespace
KuaiRec & Continue B & 280/209
& 0.0197 & 0.0214 & \textbf{0.0247} & \textbf{0.0247} \\
KuaiRec & Return to A & 991/496
& 0.0135 & 0.0152 & \textbf{0.0184} & \textbf{0.0184} \\
\bottomrule
\end{tabular}}
\end{table*}

\begin{table*}[t]
\centering
\small
\setlength{\tabcolsep}{4pt}
\caption{MATE mechanism statistics over all real test events. Long writes satisfy $w^L>0$; strong short writes satisfy $w^S>\epsilon=0.2$; valid long proxies have available supervision targets. Write evidence is measured after event revelation; actual reads $\|\alpha^r M_r z\|_2$ are measured before the prediction target is revealed. Read contributions are means over three seeds; evidence rates are shared across seeds. $N/U$ denotes events/users per seed.}
\label{tab:memory_analysis}
\resizebox{0.8\textwidth}{!}{%
\begin{tabular}{lrrrrrr}
\toprule
Dataset & $N/U$
& \shortstack{Long write\\(\%)}
& \shortstack{Strong short\\write (\%)}
& \shortstack{Valid long\\proxy (\%)}
& Long read & Short read \\
\midrule
MovieLens-10M & 913/30
& 62.3 & 11.1 & 99.9 & 17.727 & 45.927 \\
Luxury Beauty & 491/234
& 65.2 & 8.8 & 99.4 & 0.852 & 4.594 \\
KuaiRec & 12421/934
& 6.1 & 12.3 & 73.2 & 1.712 & 19.967 \\
\bottomrule
\end{tabular}}
\end{table*}

\subsection{Efficiency Analysis}
\label{sec:efficiency}

For \textbf{RQ4}, we compare Base, Single, and \textsf{MATE} using the same candidate catalog, sequence-length setting, and batch size within each dataset. For each of three seeds, we time 200 sequential positions over three repetitions after 10 warm-up events, synchronizing the accelerator at every timing boundary. We separately measure scoring before target revelation and processing after revelation, including temporal evidence and any memory update. Their sum is the total processing cost per event, not scoring-only latency. Base also computes temporal evidence, so event processing is not a pure memory-write measurement. Offline embedding generation and
initial-history replay are excluded.

Table~\ref{tab:efficiency} and Figs.~\ref{fig:efficiency_movielens}--\ref{fig:efficiency_kuairec} report the measured costs. Relative to Base, Single increases mean total cost by 18.3\%, 24.1\%, and 21.2\% on MovieLens-10M, Luxury Beauty, and KuaiRec, respectively. The corresponding increases for \textsf{MATE} are 27.1\%, 37.7\%, and 32.6\%. In absolute terms, \textsf{MATE} adds approximately 1.25--1.32~ms per event over Base and 0.42--0.48~ms over Single. Thus, the additional cost of maintaining two memories is smaller than the cost of introducing adaptation beyond Base.

Tail costs follow the same ordering: \textsf{MATE} increases P95 total cost over Base by 22.9\%, 36.9\%, and 25.6\%, respectively. Candidate-scoring time, however, changes by at most 3.0\% relative to Base. Most additional online cost therefore comes from processing revealed interactions rather than ranking candidates.

Training times include the shared warm-start stage. On MovieLens-10M, Luxury Beauty, and KuaiRec, Single increases training time over Base by 20.0\%, 23.3\%, and 13.6\%, respectively. The corresponding increases for \textsf{MATE} are 28.7\%, 45.2\%, and 22.3\%. These are descriptive three-seed measurements; the greater between-seed variability of \textsf{MATE} on Luxury Beauty limits how precisely
its training overhead can be characterized.

At float32 precision, Base has no fast-weight matrix, Single uses 0.5~MiB per active user, and \textsf{MATE} uses 1.0~MiB. Maintaining separate long- and short-term memories therefore doubles Single's adaptive state. These values exclude cached histories, temporal statistics, item embeddings, shared model weights, and temporary buffers. Peak process memory was not measured, so they represent incremental adaptive storage rather than the full deployment footprint. The state size is predictable, while online overhead is small in absolute milliseconds but non-negligible relative to Base.

\begin{table*}[t]
    \centering
    \small
    \setlength{\tabcolsep}{4pt}
    \caption{Computational and storage costs, averaged over three seeds. Scoring and event processing are mean times; P95 is the mean of
    per-seed 95th-percentile total costs. Training includes the shared
    warm-start stage. Memory counts float32 adaptive matrices only.}
    \label{tab:efficiency}
    \resizebox{0.8\textwidth}{!}{%
    \begin{tabular}{llrrrrrr}
        \toprule
        Dataset & Method
        & \shortstack{Scoring\\(ms)}
        & \shortstack{Event processing\\(ms)}
        & \shortstack{Total\\(ms)}
        & \shortstack{P95 total\\(ms)}
        & \shortstack{Training\\(h)}
        & \shortstack{Memory\\(MiB/user)} \\
        \midrule
        MovieLens-10M & Base
        & 3.377 & 1.392 & 4.769 & 5.717 & 1.479 & 0.0 \\
        MovieLens-10M & Single
        & 3.461 & 2.180 & 5.641 & 6.571 & 1.775 & 0.5 \\
        MovieLens-10M & \textsf{MATE}
        & 3.479 & 2.584 & 6.063 & 7.027 & 1.904 & 1.0 \\
        \addlinespace
        Luxury Beauty & Base
        & 3.233 & 0.268 & 3.501 & 3.745 & 0.219 & 0.0 \\
        Luxury Beauty & Single
        & 3.303 & 1.041 & 4.343 & 4.596 & 0.270 & 0.5 \\
        Luxury Beauty & \textsf{MATE}
        & 3.324 & 1.498 & 4.822 & 5.126 & 0.318 & 1.0 \\
        \addlinespace
        KuaiRec & Base
        & 3.261 & 0.555 & 3.817 & 4.423 & 1.803 & 0.0 \\
        KuaiRec & Single
        & 3.323 & 1.301 & 4.624 & 5.147 & 2.047 & 0.5 \\
        KuaiRec & \textsf{MATE}
        & 3.330 & 1.732 & 5.062 & 5.554 & 2.205 & 1.0 \\
        \bottomrule
    \end{tabular}}
\end{table*}

\begin{figure}[t]
  \centering
  \includegraphics[width=\columnwidth]{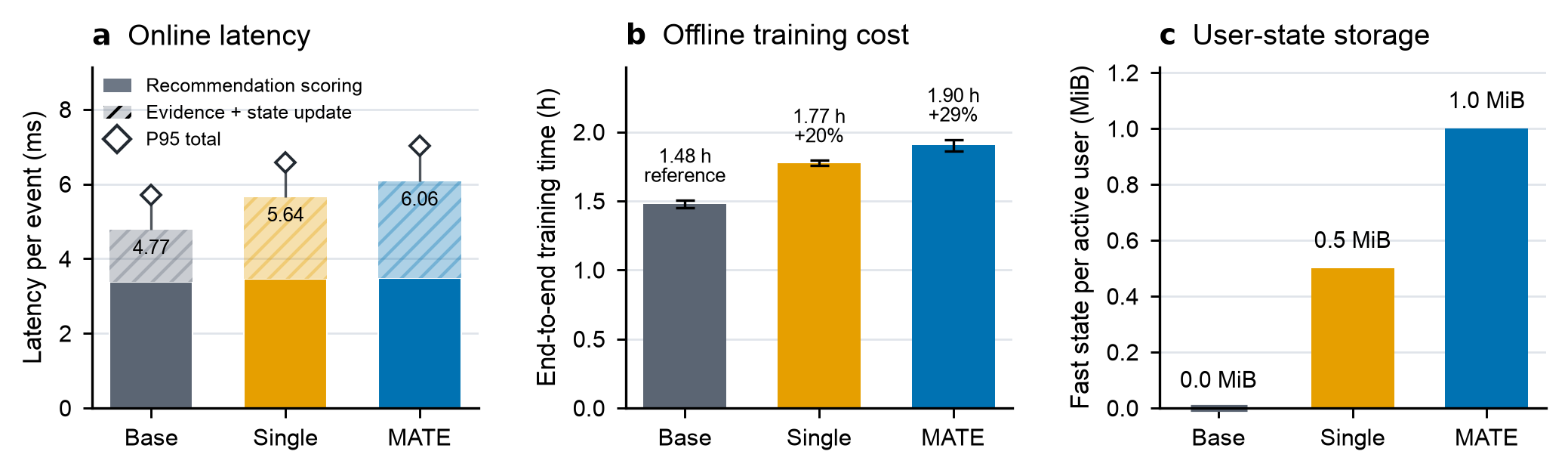}
  \Description{Computational and storage cost analysis for MovieLens-10M, including per-event latency, training time, and adaptive memory usage.}
  \caption{\textbf{Computational and storage costs on MovieLens 10M} (10,681 candidate items). (a) Mean per-event cost decomposed into candidate scoring and revealed-event processing; open diamonds show the mean of the per-seed P95 total costs. (b) End-to-end training time including the shared warm-start stage. (c) Float32 fast-weight state per active user. Bars show means over three seeds; error bars on the P95 markers and training bars show sample standard deviations. Offline item-embedding generation and initial-history replay are excluded.}
  \label{fig:efficiency_movielens}
\end{figure}

\begin{figure}[t]
  \centering
  \includegraphics[width=\columnwidth]{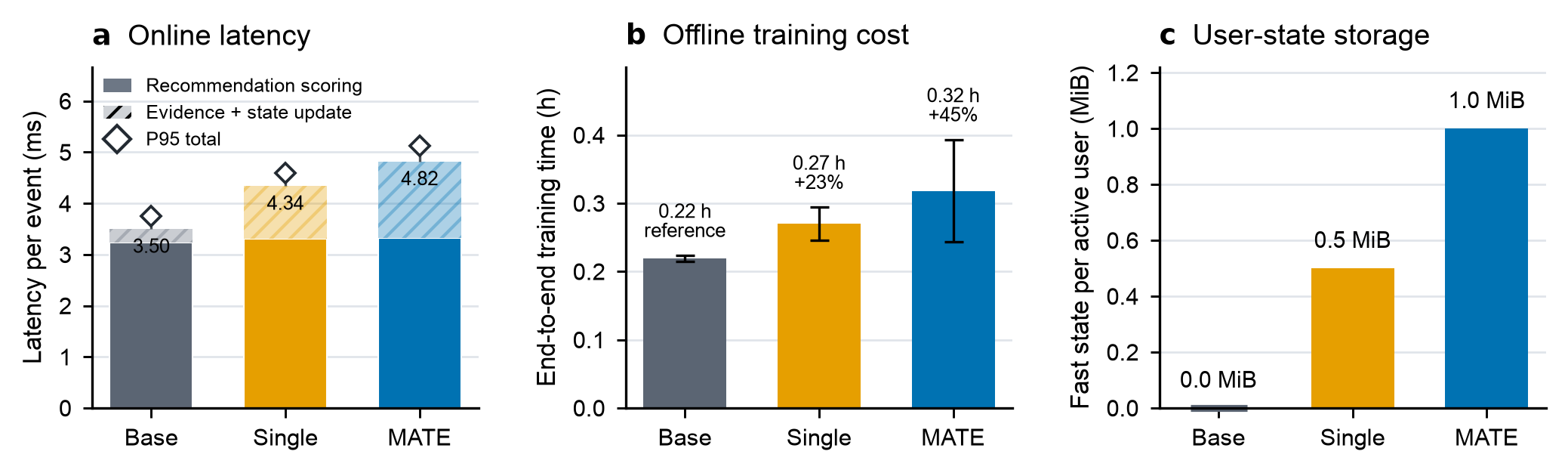}
  \Description{Computational and storage cost analysis for Luxury Beauty, including per-event latency, training time, and adaptive memory usage.}
  \caption{\textbf{Computational and storage costs on Luxury Beauty} (12,111 candidate items). (a) Mean per-event cost decomposed into candidate scoring and revealed-event processing; open diamonds show the mean of the per-seed P95 total costs. (b) End-to-end training time including the shared warm-start stage. (c) Float32 fast-weight state per active user. Bars show means over three seeds; error bars on the P95 markers and training bars show sample standard deviations. Offline item-embedding generation and initial-history replay are excluded.}
  \label{fig:efficiency_luxury}
\end{figure}

\begin{figure}[t]
  \centering
  \includegraphics[width=\columnwidth]{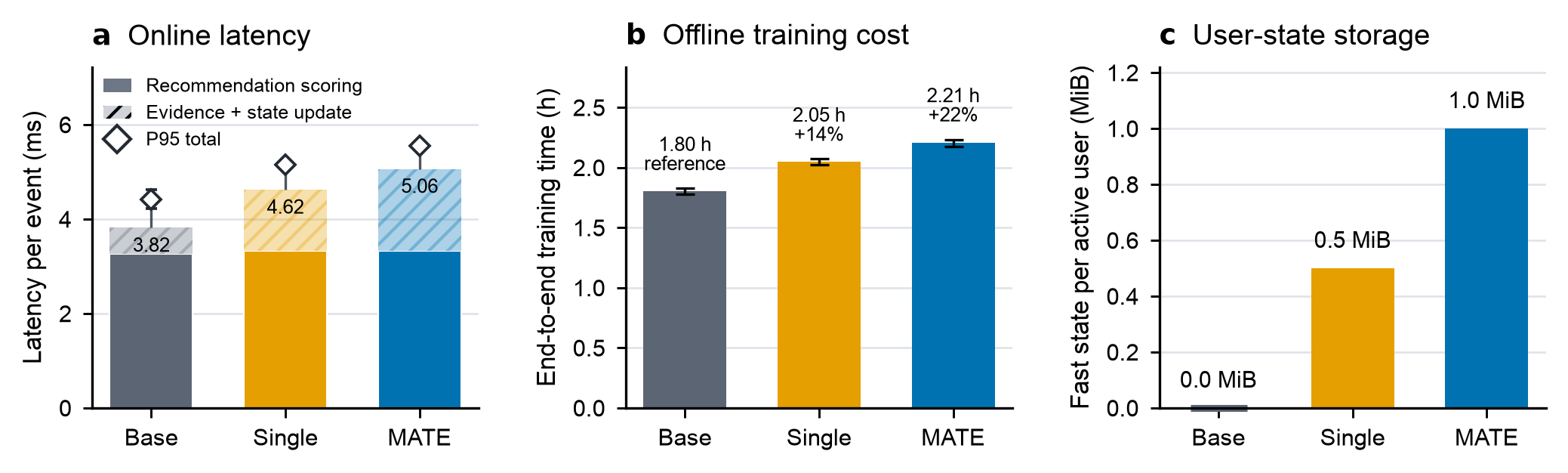}
  \Description{Computational and storage cost analysis for KuaiRec, including per-event latency, training time, and adaptive memory usage.}
  \caption{\textbf{Computational and storage costs on KuaiRec}
  (10,728 candidate items). (a) Mean per-event cost decomposed into candidate scoring and revealed-event processing; open diamonds show the mean of the per-seed P95 total costs. (b) End-to-end training time including the shared warm-start stage. (c) Float32 fast-weight state per active user. Bars show means over three seeds; error bars on the P95 markers and training bars show sample standard deviations. Offline item-embedding generation and initial-history replay are excluded.}
  \label{fig:efficiency_kuairec}
\end{figure}

\subsection{Qualitative User Understanding}
\label{sec:case_study}

For \textbf{RQ5}, we present one real Luxury Beauty example that makes the memory mechanism directly inspectable. This is an outcome-conditioned positive case selected after evaluation for readability; it was not used for checkpoint or hyperparameter selection and is not an estimate of how often \textsf{MATE} improves recommendations. The user identifier is anonymized, while item IDs, dates, ratings, and saved predictions remain traceable to the experiment records.

Table~\ref{tab:case_study} contains the observed prefix, write evidence, subsequent memory read, and recommendations in one place. The prefix contains an OPI thinner interaction and eight essie polish interactions in 2016, followed by another observed OPI polish in 2018. After that last event was revealed, similar behavior was supported in both historical months across a 540-day span. It therefore received a full long-term write ($w^L=1.0$). No matching behavior occurred in the seven-day short window ($C=0.0$), so the short memory received its minimum write ($w^S=0.2$).

At the following prediction, both memory gates were active ($\alpha^L=0.619$, $\alpha^S=0.567$). The actual read contributions were $\lVert\alpha^L M_L k\rVert_2=0.157$ and $\lVert\alpha^S M_S k\rVert_2=3.242$; hence the short-memory contribution was 20.7 times the long-memory contribution. Thus, a nonzero long-term read coexists with a much larger short-term read at this position. This is a mechanism diagnostic rather than proof that either branch alone caused the ranking change.

The real next item was another OPI polish. Under the same 12,111-item catalog, identical history masking, and the same 10-event prefix, its rank improves from 18 for Base and 3 for Single to 2 for \textsf{MATE}. \textsf{MATE}'s Top-3 also contains two OPI products, with the real target at position~2. Frozen Online Memory ties \textsf{MATE} at rank~2; therefore, this example shows that the dual-memory \textsf{MATE} configuration ranks the target above Base and Single, but it does not isolate an additional discrete-rank benefit from the latest test-time write.

All ranks and recommendation lists were saved before event~10 was revealed. The write weights computed from that target are excluded and can affect only later predictions. Events~9 and~10 share the same date-level timestamp, so their sequence follows the recorded event IDs; the data do not establish a within-day clock-time order.

\begin{table*}[t]
  \centering
  \small
  \setlength{\tabcolsep}{4pt}
  \caption{One positive real-user case from Luxury Beauty (seed~42).
  All methods use the same 12,111-item catalog and history
  mask at the same observed prefix before the real target is revealed. Lower
  target rank is better; a bold title marks the real target when it appears in
  a Top-3. The case is a mechanism illustration, not a representative
  performance estimate.}
  \label{tab:case_study}
  \begin{tabularx}{\textwidth}{p{1.35cm} p{2.05cm} X p{3.25cm}}
    \toprule
    Phase & Date/method & Observed interaction or saved Top-3 & Evidence/result \\
    \midrule
    Observed & 2016-10-10
      & \texttt{LUX18\_B000PZ8E62}: \emph{OPI Nail Lacquer Thinner, 2 Fl Oz}
      & Rating 5.0 \\
    Observed & 2016-11-18
      & Eight distinct \emph{essie Nail Color Polish} items
      & Eight interactions, all rating 5.0 \\
    Observed & 2018-04-03
      & \texttt{LUX18\_B000142FVW}: \emph{OPI Nail Lacquer, Not So Bora-Bora-ing Pink, 0.5 Fl Oz}
      & Rating 5.0; after reveal:
        $w^L=1.0$, $w^S=0.2$;
        2/2 matching months \\
    \addlinespace
    MATE read & Before target
      & $\alpha^L=0.619$,
        $\alpha^S=0.567$; both proxies valid
      & Contributions: L=0.157,
        S=3.242 \\
    \addlinespace
    Prediction & Base
      & 1. essie Nail Color Polish; 2. essie Nail Color Polish; 3. essie Nail Color Polish
      & Target rank 18 \\
    Prediction & Single
      & 1. essie Nail Color Polish; 2. essie Nail Color Polish; 3. \textbf{OPI Nail Lacquer, Cajun Shrimp, 0.5 fl. oz.}
      & Target rank 3 \\
    Prediction & Frozen Online Memory
      & 1. essie Nail Color Polish; 2. \textbf{OPI Nail Lacquer, Cajun Shrimp, 0.5 fl. oz.}; 3. OPI Nail Lacquer, My Very First Knockwurst, 0.5 fl. oz.
      & \textbf{Target rank 2} \\
    Prediction & MATE
      & 1. essie Nail Color Polish; 2. \textbf{OPI Nail Lacquer, Cajun Shrimp, 0.5 fl. oz.}; 3. OPI Nail Lacquer, My Very First Knockwurst, 0.5 fl. oz.
      & \textbf{Target rank 2} \\
    \addlinespace
    Target & 2018-04-03
      & \texttt{LUX18\_B000NG80GM}: \emph{OPI Nail Lacquer, Cajun Shrimp, 0.5 fl. oz.}
      & Rating 5.0; revealed only after ranking \\
    \bottomrule
  \end{tabularx}
\end{table*}

\section{Conclusion and Broad Impact}
\label{sec:conclusion}

In this work, we studied evolving user interests in LLM-enhanced sequential recommendation with frozen item embeddings. We proposed \textsf{\textbf{MATE}}, a memory adaptation framework that uses temporal evidence to update two user-specific memories with different adaptation dynamics. Long-term evidence captures recurring behavior across periods, while short-term evidence reflects recent consistency. Context-dependent reading determines their contributions to the current recommendation. The shared model is trained offline, and user memories are updated only after real interactions become observable.

Experiments on MovieLens-10M, Amazon Luxury Beauty, and KuaiRec show that \textsf{MATE} achieves the best mean performance across the reported recommendation metrics. Equal-capacity ablations support the roles of temporal evidence, temporal supervision, and dynamic reading. Real-sequence analyses further show improved prediction of both continued recent interests and actual returns to earlier recurring themes. However, Frozen Online Memory matches \textsf{MATE}'s reported Top-$K$ metrics, so these experiments do not establish an additional benefit from continued test-time writing. Efficiency measurements show that adaptation introduces additional computation and user-state storage. Together, the results support temporal evidence as a useful control signal for memory adaptation, while leaving its incremental online benefit open to further evaluation.

\paragraph{Broad Impact}
MATE offers a framework for personalized recommendation that combines persistent information with adaptation to recent behavior. This design may also inform other interactive systems with evolving user states. Its deployment nevertheless requires careful handling of behavioral histories and persistent memories. Privacy protection, memory management, and safeguards against reinforcing narrow or outdated interests remain important considerations. Predictive improvements should not be interpreted as direct understanding of users' motivations, and practical evaluation should consider user benefit alongside recommendation accuracy and computational cost.

\appendix

\section{Detailed Theoretical Analysis}
\label{app:mate_theory}

\begin{proof}[Proof of Theorem~\ref{thm:selective_writing}]
Let
$N=N^L_{u,t}=\sum_{m\in\mathcal{M}_{u,t}}z^m_{u,t}$.
Removing the time-span condition can only increase the
long-term write evidence. Hence,
\begin{equation}
    0\leq w^L_{u,t}
    \leq
    \frac{N}{n}\mathbb{I}[N\geq2].
\end{equation}

By independence and the upper bound on each monthly support
probability, $N$ is stochastically dominated by
$\widetilde N\sim\operatorname{Binomial}(n,p)$.
This follows by coupling the indicators using independent uniform
random variables and increasing each success probability to $p$.
Since $x\mapsto(x/n)\mathbb{I}[x\geq2]$ is nondecreasing on
the nonnegative integers,
\begin{align}
    \mathbb{E}[w^L_{u,t}]
    &\leq
    \mathbb{E}\left[
        \frac{\widetilde N}{n}
        \mathbb{I}[\widetilde N\geq2]
    \right]\\
    &=
    \frac{\mathbb{E}[\widetilde N]}{n}
    -
    \frac{\Pr(\widetilde N=1)}{n}\\
    &=
    p-p(1-p)^{n-1}.
\end{align}

Finally, applying
$1-(1-p)^{n-1}\leq(n-1)p$
gives
\begin{equation}
    \mathbb{E}[w^L_{u,t}]
    \leq
    p\bigl[1-(1-p)^{n-1}\bigr]
    \leq
    (n-1)p^2.
\end{equation}
The additional time-span condition cannot invalidate this upper
bound because it only suppresses otherwise eligible writes.
\end{proof}

\begin{proof}[Proof of Theorem~\ref{thm:memory_stability}]
For $a>0$, define radial clipping by
\begin{equation}
    \Pi_a(\mathbf X)
    =
    \mathbf X
    \min\left\{1,\frac{a}{\|\mathbf X\|_F}\right\},
    \qquad
    \Pi_a(\mathbf 0)=\mathbf 0.
\end{equation}
This is projection onto the Frobenius-norm ball of radius $a$.
Projection optimality and the Cauchy--Schwarz inequality give
\begin{align}
    &\|\Pi_a(\mathbf X)-\Pi_a(\mathbf Y)\|_F^2
    \notag\\
    &\quad\leq
    \langle
        \Pi_a(\mathbf X)-\Pi_a(\mathbf Y),
        \mathbf X-\mathbf Y
    \rangle_F
    \notag\\
    &\quad\leq
    \|\Pi_a(\mathbf X)-\Pi_a(\mathbf Y)\|_F
    \|\mathbf X-\mathbf Y\|_F.
\end{align}
Thus, clipping is nonexpansive:
\begin{equation}
    \|\Pi_a(\mathbf X)-\Pi_a(\mathbf Y)\|_F
    \leq
    \|\mathbf X-\mathbf Y\|_F.
\end{equation}

The following argument applies to either memory.
Fix $r=L$ or $r=S$.
Including the clipping operations, the update is
\begin{equation}
    \mathbf M^r_{u,j}
    =
    \Pi_B\left(
        \rho_r\mathbf M^r_{u,j-1}
        +
        \Pi_c\left(
            \eta_rw^r_{u,j}
            \mathbf e_{u,j}
            \mathbf k_{u,j-1}^{\top}
        \right)
    \right).
\end{equation}

At step $t$, the reference retains
$\Pi_B(\rho_r\mathbf M^r_{u,t-1})$
but omits the additive write.
Since $\|\mathbf e_{u,t}\|_2=1$, nonexpansiveness gives
\begin{align}
    &\|\mathbf M^r_{u,t}
       -\widetilde{\mathbf M}^r_{u,t}\|_F
    \notag\\
    &\quad\leq
    \left\|
        \Pi_c\left(
            \eta_rw^r_{u,t}
            \mathbf e_{u,t}
            \mathbf k_{u,t-1}^{\top}
        \right)
    \right\|_F
    \notag\\
    &\quad\leq
    \min\left\{
        c,\eta_rw^r_{u,t}
        \|\mathbf k_{u,t-1}\|_2
    \right\}.
\end{align}

For every later step $j>t$, the two trajectories receive identical
write increments. These increments cancel inside the projection
comparison, yielding
\begin{equation}
\begin{aligned}
    &\|\mathbf M^r_{u,j}
       -\widetilde{\mathbf M}^r_{u,j}\|_F\\
    &\quad\leq
    \rho_r
    \|\mathbf M^r_{u,j-1}
       -\widetilde{\mathbf M}^r_{u,j-1}\|_F.
\end{aligned}
\end{equation}
Iterating for $s$ steps proves
Eqs.~\eqref{eq:long_influence_bound}
and~\eqref{eq:short_influence_bound}.

When clipping is inactive, the difference is exactly
\begin{equation}
    \mathbf M^r_{u,t+s}
    -
    \widetilde{\mathbf M}^r_{u,t+s}
    =
    \rho_r^s\eta_rw^r_{u,t}
    \mathbf e_{u,t}\mathbf k_{u,t-1}^{\top}.
\end{equation}
With clipping, every updated memory additionally satisfies
\begin{equation}
    \|\mathbf M^r_{u,j}\|_F\leq B.
\end{equation}
\end{proof}

\bibliographystyle{ACM-Reference-Format}
\bibliography{references}

\end{document}